\documentclass{article}
\usepackage{amssymb}

\usepackage{makeidx}
\usepackage{graphicx}
\usepackage{amsmath,mathrsfs}

\newtheorem{theorem}{Theorem}

\newtheorem{corollary}[theorem]{Corollary}

\newtheorem{definition}[theorem]{Definition}

\newtheorem{lemma}[theorem]{Lemma}

\newtheorem{proposition}[theorem]{Proposition}

\newenvironment{proof}[1][Proof]{\textbf{#1.} }{\ \rule{0.5em}{0.5em}}

\begin{document}

\title{Conditioning of Quantum Open Systems}
\author{John Gough\\
Aberystwyth University\\    
SY23 3BZ, Aberystwyth, Wales, UK\\
E-mail: jug@aber.ac.uk}
\date{}
\maketitle

\begin{abstract}
The underlying probabilistic theory for quantum mechanics is non-Kolmogorovian. The order in which physical observables will be important
if they are incompatible (non-commuting). In particular, the notion of conditioning needs to be handled with care and may not even exist in some cases. Here we layout the quantum probabilistic formulation in terms of von Neumann algebras, and outline conditions (non-demolition properties) under which filtering may occur.
\end{abstract}

\section{Introduction}
The mathematical theory of quantum probability (QP) is an extension of the usual Kolmogorov probability to the setting inspired by quantum theory \cite{Maassen,AFL,HP84,Par92}. In this article, we will emphasize the quantum analogues to events (projections), random variables (operators), sigma-algebras (von Neumann algebras), probabilities (states), etc.  
The departure from Kolmogorov's theory is already implicit in the fact that
the quantum random variables do not commute. It is advantageous to set up a
non-commutative theory of probability and one of the requirements is that
Kolmogorov's theory is contained as the special case when we restrict to
commuting observables. The natural analogue of measure theory for operators
is the setting of von Neumann algebras \cite{KR}. 

\section{Keywords}
Quantum Filter, von Neumann algebra, quantum probability, Bell's Theorem

\section{Classical Probability}

The standard approach to probability in the classical world is to associate
with each model a Kolmogorov triple $\left( \Omega ,\mathcal{F},\mathbb{P}%
\right) $ comprising a measurable space $\Omega $, a sigma-algebra, $%
\mathcal{F}$, of subsets of $\Omega $, and probability measure $\mathbb{P}$
on the sigma-algebra. $\Omega $ covers the entirety of all possible things
that may enfold in our model, the elements of $\mathcal{F}$ are
distinguished subsets of $\Omega $ referred to as events, and $\mathbb{P}%
\left[ A\right] $ is the probability of event $A$ occurring. From a practical point of view we need to frame the model in terms
of what we may hope to observe, and these constitute the ``events'', $\mathcal{F}$, and from a mathematical point of view, restricting to a sigma-algebra
clarifies all the ambiguities, while also resolving all the technical
issues, pathologies, etc., that would otherwise plague the subject.

A \textbf{random variable} is then defined as a function on $\Omega $ to a value
space, another measurable space $\left( \Gamma ,\mathcal{G}\right) $: so, for each $G\in \mathcal{G}$, the set $X^{-1}\left[
G\right] =\left\{ \omega \in \Omega :X\left( \omega \right) \in G\right\} $
is an event, i.e. in $\mathcal{F}$.  The probability that $X$ takes a value in $G$ is then $%
\mathbb{K}_{X}\left[ G\right] =\mathbb{P}\left[ X^{-1}\left[ G\right] \right]
$ which produces the distribution of $X$ as $\mathbb{K}_{X}=\mathbb{P\circ X}^{-1}$.
Let $X_{1},\cdots ,X_{n}$ be random variables
then their joint probabilities are also well-defined: 
\begin{equation}
\mathbb{K}_{X_{1},\cdots ,X_{n}}\left[ G_{1},\cdots ,G_{n}\right] =\mathbb{P}%
\left[ X_{1}^{-1}\left[ G_{1}\right] \cap \cdots \cap X_{n}^{-1}\left[ G_{n}%
\right] \right] .
\end{equation}

Let $\mathcal{A}$ be a sub-sigma algebra of $\mathcal{F}$. The collection of bounded $\mathcal{A}$-measurable (complex-valued) functions will be denoted as $\mathscr{A}=L^{\infty }(\Omega ,\mathcal{A})$. This is an example of a *-algebra of functions.
We now show that there is a natural identification between $\sigma$-algebras of subsets of $\Omega$ and *-algebras of functions on $\Omega$. First we need some definitions.

\begin{definition}
\index{Positive bounded monotone sequence} A sequence $\left( f_{n}\right)
_{n}$ of functions on $\Omega$ is said to be be a \textbf{positive bounded monotone
sequence} if there exists a finite constant $c>0$ such that $0\leq f_{n}\leq
c$ and $f_{n}\leq f_{n+1}$ for each $n$.
A *-algebra of functions $\mathscr{A}$ on $\Omega$ is said to be \textbf{monotone class} if every positive bounded monotone sequence in $\mathscr{A}$ has its limit in $\mathscr{A}$.
\end{definition}

The next result can be found in Protter \cite{Protter}.

\begin{theorem}[Monotone Class Theorem]
\index{Theorem!Monotone Class}
Given a monotone class *-algebra of functions, $\mathscr{A}$, on $\Omega$. Then $\mathscr{A}=L^\infty ( \Omega , \mathcal{A})$ where $\mathcal{A}$ is precisely the $\sigma$-algebra generated by $\mathscr{A}$ itself.
\end{theorem}

Conditional probabilities are natural defined: the probability that event $A$ occurs
given that event $B$ occurs is 
\begin{equation}
\mathbb{P}\left[ A|B\right] \triangleq \frac{\mathbb{P}\left[ A\cap B\right] 
}{\mathbb{P}\left[ B\right] }.
\end{equation}
This is a simple ``renormalization'' of the probability: one restricts the
outcomes in $A$ which also lie in $B$, weighted as a proportion out of all $%
B $, rather than $\Omega $.

\subsection{Quantum Probability Models}

The standard presentation of quantum mechanics takes physical quantities
(observables) to be self-adjoint operators on a fixed Hilbert space $%
\mathfrak{h}$ of wavefunctions. The normalized elements of $\mathfrak{h}$
are the \textbf{pure states}, and the expectation of an observable $X$ for a
pure state $\psi \in \mathfrak{h}$ is given $\mathbb{E}\left[ X\right] =\langle \psi |\,X\psi \rangle $.
As such, observables play the role of random variables. More generally, we
encounter quantum expectations of the form 
\begin{equation}
\mathbb{E}\left[ X\right] =tr\left\{ \rho X\right\} 
\label{eq:density_matrix_expectation}
\end{equation}
where $\rho \geq 0$ is a trace-class operator normalized so that $tr\left\{
\rho \right\} =1$. The operator $\rho $ is called a \textbf{density matrix}
and in the pure case corresponds to $\rho =|\psi \rangle \langle \psi |$.

\begin{definition}
A \textbf{quantum probability space} $\left( \mathfrak{A},\mathbb{E}\right) $
consists of a von Neumann algebra $\mathfrak{A}$ and a state $\mathbb{E}$
(assumed to be continuous in the normal topology). 
\end{definition}

When $\mathfrak{A}$ is commutative, then the quantum
probability space is isomorphic to a Kolmogorov model. 
Let us motivate now why von Neumann algebras are the appropriate object. Positive operators are well defined, and we may say $X \le Y$ if $Y-X \ge0$. In particular, the concept of a positive bounded
monotone sequence of operators makes sense, as does a monotone class algebra of operators. 

\begin{theorem}[van Handel \cite{vH_thesis}]
A collection of bounded operators over a fixed Hilbert space, is a von Neumann algebra if and only if it is a monotone class *-algebra.
\end{theorem}

Specifically, we see that this recovers the usual monotone class Theorem when we further impose commutativity of the algebra.
A state on a von Neumann algebra, $\mathfrak{A}$, is then a normalized
positive linear map from $\mathfrak{A}$ to the complex numbers, that is, 
$\mathbb{E}[\hbox{1\kern-4truept 1}]=1,\quad \mathbb{E}[\alpha A+\beta
B]=\alpha \,\mathbb{E}[A]+\beta \,\mathbb{E}[B],\quad \mathbb{E}[X]\geq 0$, whenever $X\geq 0$.
Note that if $(X_n)$ is a positive bounded monotone sequence with limit $X$, then $\mathbb{E}[X_n]$ converges to $\mathbb{E}[X]$ - this in fact equivalent to the condition of continuity in the normal topology, and implies that $\mathbb{E}$ to take the form (\ref{eq:density_matrix_expectation}), for some density
matrix $\rho $, \cite{KR}.

\begin{definition}
An observable is referred to as a \textbf{quantum event} if it is an
orthogonal projection. If a quantum event $A$ corresponds to a projection $%
P_{A}$ then its probability of occurring is $\Pr \left\{ A\right\} =tr\left\{ \rho P _{A}\right\} $.
\end{definition}

The requirement that $P_{A}$ is an orthogonal projection means that $%
P_{A}^{2}=P_{A}=P_{A}^{\dag }$. The complement to the event $A$, that is 
\textit{not} $A$, will be denoted as $\bar{A}$, and is the orthogonal
projection given by the orthocomplement $P_{\bar{A}}=\hbox{1\kern-4truept 1}%
-P_{A}$, where $\hbox{1\kern-4truept 1}$ is the identity operator on $%
\mathfrak{h}$. A von Neumann algebra is a subalgebra of the bounded
operators on $\mathfrak{h}$ with good closure properties: crucially it will
be generated by its projections. So the von Neumann algebra generated by a
collection of quantum events is the natural analogue of a sigma-algebra of
classical events.

However, the new theory of quantum probability will have features not
present classically. For instance, the notion of a pair of events, $A$ and $%
B $, occurring jointly is not generally meaningful. In fact, $P _{A}P _{B}$
is in general not self-adjoint, and so does not correspond to a quantum
event. We therefore cannot interpret $tr\{\rho P _{A}P _{B}\}$ as the joint
probability for quantum events $A$ and $B$ to occur.

\begin{proposition}
The product $P_A P_B$ is an event (that is, an orthogonal projection) if and
only if $P_A$ and $P_B$ commute.
\end{proposition}

\begin{proof}
Self-adjointness of $P_A P_B$ is enough to give the commutativity since then 
$P_A P_B =( P_A P_B )^\ast =P_B P_A$ as both $P_A$ and $P_B$ are
self-adjoint. We then see that $(P_A P_B )^2 = P_A P_B P_A P_B = P_A^2 P_B^2
= P_A P_B$.
\end{proof}

As such, given a pair of quantum events $A$ and $B$, it is usually
meaningless to speak of their joint probability $\Pr \{ A,B \}$ for a fixed
state $\rho$. An exception is made when the corresponding projectors commute
in which case we say the events are compatible and take the probability to
be $\Pr \{ A,B \} \equiv tr \{ \rho P_A P_B \}$.

By the spectral theorem, every observable may be written as 
\begin{equation}
X=\int_{-\infty }^{\infty }x\,P_{X}[dx]
\end{equation}
where $P_{X}[dx]$ is a projection valued measure, normalized so that $P_{X}%
\left[ \mathbb{R}\right] =1$, the identity operator. The measure is
supported on the spectrum of $X$ which is, of course real by
self-adjointness. In particular, if $G_{1}$ and $G_{2}$ are non-overlapping
Borel subsets of $\mathbb{R}$ then $P_{X}[G_{1}]$ and $P_{X}[G_{2}]$ project
onto mutually orthogonal projections.

The orthogonal projection $P_{X}[G]$ then corresponds to the quantum event
that $X$ is measured to have a value in the subset $G$. The smallest von
Neumann algebra containing all the projections $P_{X}[G]$ will be denoted as 
$\mathfrak{F}_{X}$ and plays an analogous role to the sigma-algebra
generated by a random variable.

Once we fix the density matrix $\rho$, the spectral decomposition leads to
the probability distribution of observable $X$: 
$\mathbb{K}_{X}\left[ dx\right] =tr\left\{ \rho P _{X}\left[ dx\right]
\right\} $.
We say that observables $X_{1},\cdots ,X_{n}$ are \textbf{compatible} if the
quantum events they generate are compatible. In this case 
\begin{eqnarray}
tr\left\{ \rho e^{i\sum_{k=1}^{n}u_{k}X_{k}}\right\} \equiv \int
e^{i\sum_{k=1}^{n}u_{k}x_{k}}\mathbb{K}_{X_{1},\cdots ,X_{n}}\left[
dx_{1},\cdots dx_{n}\right] ,
\end{eqnarray}
where $\mathbb{K}_{X_{1},\cdots ,X_{n}}\left[ dx_{1},\cdots dx_{n}\right] $
defines a probability measure on the Borel sets of $\mathbb{R}^{n}$. This
may be no longer true if we drop the compatibility assumption!

We remark that, given a collection of observables $X_{1},\cdots ,X_{n}$, we
can construct the smallest von Neumann algebra containing all their
individual quantum events; this will typically a non-commutative algebra and
effectively plays the role of a sigma algebra generated by random variables.

Positivity preserving measurable mappings are the natural morphisms between
Kolmogorov spaces. The situation in quantum probability is rather more
prescriptive.

First note that if $\mathfrak{A}$ and $\mathfrak{B}$ are von Neumann
algebras, then so too is their formal tensor product. \ A map $\Phi :%
\mathfrak{F}\rightarrow \mathfrak{F}$ between von Neumann algebras is
positive is $\Phi \left( A\right) \geq 0$ whenever $A\geq 0$. However, we
need a stronger condition. The mapping has the extension $\tilde{\Phi}_{%
\mathfrak{F}}:\mathfrak{A}\otimes \mathfrak{F}\rightarrow \mathfrak{B}\otimes %
\mathfrak{F}$ for a given von Neumann algebra by $\tilde{\Phi}_{\mathfrak{F}%
}\left( A\otimes F\right) =\Phi \left( A\right) \otimes F$. We say that $%
\Phi $ is \textbf{completely positive (CP)} if $\tilde{\Phi}_{\mathfrak{F}}$
is positive for any $\mathfrak{F}$.

A \textbf{morphism between a pair of quantum probability spaces} \cite{Maassen}, $\Phi : (\mathfrak{A}_1, \mathbb{E}_1) \mapsto (\mathfrak{A}_2, \mathbb{E}_2)$ is a
completely positive map with the properties $\Phi( \hbox{1\kern-4truept 1}_{\mathfrak{A}_1} )
=\hbox{1\kern-4truept 1}_{\mathfrak{A}_2}$ and $\mathbb{E}_2 \circ \Phi = \mathbb{E}_1$.
Despite its rather trivial looking appearance, the CP property is
actually quite restrictive.

\section{Quantum Conditioning}

The conditional probability of event $A$ occurring given that $B$ occurs is
defined by 
$\Pr \left\{ A|B\right\} =\frac{\Pr \left\{ A,B\right\} }{P\left\{ B\right\} }$.
In quantum probability, $\Pr \left\{ A,B\right\}$ may make sense as a joint
probability only if $A$ and $B$ are compatible, otherwise there is the
restriction that $A$ is measured before $B$.

Let $X$ be an observable with a discrete spectrum (eigenvalues). If we start
in a pure state $|\psi _{\text{in}}\rangle $ and measure $X$ to record a
value $x$ then this quantum event has corresponding projector $P _{X}\left(
x\right) $. Von Neumann's projection postulate states that the state after
measurement is proportional to
\begin{eqnarray}
|\psi _{\text{out}}\rangle =P _{X}\left( x\right) |\psi _{\text{in}}\rangle
\end{eqnarray}
and that the probability of this event is $\Pr \left\{ X=x\right\} =\langle
\psi _{\text{in}}|P _{X}\left( x\right) |\psi _{\text{in}}\rangle \equiv
\langle \psi _{\text{out}}|\psi _{\text{out}}\rangle $. Note that $|\psi _{%
\text{out}}\rangle $ is not normalized! 
A subsequent measurement of another discrete observable $Y$, leading to
eigenvalue $y$, will result in $|\psi _{\text{out}}\rangle =P _{Y}\left(
y\right) P _{X}\left( x\right) |\psi _{\text{in}}\rangle $ and so (ignoring
dynamics form the time being) 
\begin{eqnarray}
\Pr \left\{ X=x;Y=y\right\} =\langle \psi _{\text{out}}|\psi _{\text{out}%
}\rangle =\langle \psi _{\text{in}}|P _{Y}\left( y\right) P _{X}\left(
x\right) P _{Y}\left( y\right) |\psi _{\text{in}}\rangle .
\end{eqnarray}
This needs to be interpreted as a sequential probability - event $X=x$
occurs first, and $Y=y$ second - rather than a joint probability. If $X$ and 
$Y$ do not commute then the order in which they are measured matters as $\Pr
\left\{ X=y;Y=y\right\} $ may differ from $\Pr \left\{ Y=y;X=x\right\} $.

\begin{lemma}
Let $P$ and $Q$ be orthogonal projections then the properties $QPQ=PQP$ and $%
PQP=QP$ are equivalent to $\left[ Q,P\right] =0$.
\end{lemma}

\begin{proof}
We begin by noting that $\left[ P,Q\right] ^{\ast }\left[ P,Q\right] =QPQ-PQPQ+PQP-QPQP$.
This may be rewritten as $\left( Q-P\right) \left( QPQ-PQP\right) $ so if $%
QPQ=PQP$ then $\left[ P,Q\right] =0$. If we have $PQP=QP$, then $\left[ P,Q%
\right] ^{\ast }\left[ P,Q\right] \equiv QPQ-QPQ+PQP-QP=PQP-QP=0$ and so $%
\left[ P,Q\right] =0$. Conversely, if $P$ and $Q$ commute then $QPQ=PQP$ and 
$PQP=QP$ hold true.
\end{proof}

\begin{corollary}
If we have $\Pr \left\{ X=x;Y=y\right\}$ equal to $\Pr \left\{
Y=y;X=x\right\} $ for all states $|\psi _{\text{in}}\rangle $ and all
eigenvalues $x$ and $y $are equal, then $X$ and $Y$ are compatible.
\end{corollary}

The symmetry implies that $P _{Y}\left( y\right) P _{X}\left( x\right) P
_{Y}\left( y\right)$ equals $P _{X}\left( x\right) P _{Y}\left( y\right) P
_{X}\left( x\right) $ and so the spectral projections of $X$ and $Y$ commute.

\begin{corollary}
If for all states $|\psi _{\text{in}}\rangle $ whenever we measure $X$, then 
$Y$ and then $X$ again, we always record the same value for $X$, then $X$
and $Y$ are compatible.
\end{corollary}

Setting $|\psi _{\text{out}}\rangle =P _{Y}\left( y\right) P _{X}\left(
x\right) |\psi _{\text{in}}\rangle $ we must have that $P _{X}\left(
x\right) |\psi _{\text{out}}\rangle =|\psi _{\text{out}}\rangle $, if this
is true for all $|\psi _{\text{in}}\rangle $ then $P _{X}\left( x\right) P
_{Y}\left( y\right) =P _{X}\left( x\right) P _{Y}\left( y\right) P
_{X}\left( x\right) $ which likewise implies that $\left[ P _{X}\left(
x\right) ,P _{Y}\left( y\right) \right] =0$.

\subsection{Conditioning Over Time}

The dynamics under a (possibly time-dependent) Hamiltonian $H\left( t\right) 
$ is described by the two-parameter family of unitary operators 
\begin{eqnarray}
U\left( t,s\right) =\hbox{1\kern-4truept 1}-i\int_{s}^{t}H\left( \tau
\right) U\left( \tau ,s\right) d\tau ,\quad t\geq s.
\end{eqnarray}
This is the solution to $\frac{\partial }{\partial t}U\left( t,s\right)
=-iH\left( t\right) U\left( t,s\right) $, $U\left( s^{-},s\right) =%
\hbox{1\kern-4truept 1}$. We have the \textbf{flow identity} 
\begin{eqnarray}
U\left( t_{3},t_{2}\right) U\left( t_{2},t_{1}\right) =U\left(
t_{3},t_{1}\right)  \label{eq:flow}
\end{eqnarray}
whenever $t_{3}\geq t_{2}\geq t_{1}$. In the special case where $H$ is
constant, we have $U\left( t,s\right) =e^{-i\left( t-s\right) H}$.
It is convenient to introduce the maps $\left( t_{2}>t_{1}\right) $%
\begin{eqnarray}
J_{\left( t_{1},t_{2}\right) }\left( X\right) =U\left( t_{2},t_{1}\right)
^{\ast }XU\left( t_{2},t_{1}\right)   
\end{eqnarray}
and, from the flow identity (\ref{eq:flow}), 
$J_{\left( t_{1},t_{2}\right) }\circ J_{\left( t_{2},t_{3}\right) }
=J_{\left( t_{1},t_{3}\right) }$ whenever $t_{3}>t_{2}>t_{1}$.

We now consider an experiment where we measure observables $Z_{1},\cdots
,Z_{n}$ at times $t_{1}<t_{2}<\cdots <t_{n}$ during a time interval $0$ to $%
T $. At the end of the experiment, if we measure $Z_{1}\in G_{1},\cdots
,Z_{n}\in G_{n}$, then the output state should be 
\begin{eqnarray}
|\psi _{\text{out}}\rangle =U\left( T,t_{n}\right) P _{Z_{n}}\left[ G_{n}%
\right] U\left( t_{n},t_{n-1}\right) \cdots U\left( t_{2},t_{1}\right) P
_{Z_{1}}\left[ G_{1}\right] U\left( t_{1},0\right) |\psi _{\text{in}}\rangle
.
\end{eqnarray}
It is convenient to introduce the observables 
\begin{eqnarray}
Y_{k}=J_{\left( 0,t_{k}\right) }\left( Z_{k}\right) .
\end{eqnarray}
The quantum event $Y_{k}\in G_{k}$ at time $t_k$ is then $P _{Y_{k}}\left[
G_{k}\right] =J_{\left( 0,t_{k}\right) }\left( P _{Z_{k}}\left[ G_{k}\right]
\right) $. The flow identity implies $U\left( t_{k},t_{k-1}\right) =U\left(
t_{k},0\right) U\left( t_{k-1},0\right) ^{\ast }$, and we find 
\begin{eqnarray}
|\psi _{\text{out}}\rangle &=&U\left( T,0\right) J_{\left( 0,t_{n}\right)
}\left( P _{Z_{n}}\left[ G_{n}\right] \right) \cdots J_{\left(
0,t_{n}\right) }\left( P _{Z_{1}}\left[ G_{1}\right] \right) |\psi _{\text{in%
}}\rangle  \notag \\
&=&U\left( T,0\right) P _{Y_{n}}\left[ G_{n}\right] \cdots P _{Y_{1}}\left[
G_{1}\right] |\psi _{\text{in}}\rangle .
\end{eqnarray}
We therefore have the probability 
\begin{eqnarray}
\Pr \left\{ Y_{1}\in G_{1};\cdots ;Y_{n}\in G_{n}\right\}=\langle \psi _{%
\text{out}}|\psi _{\text{out}}\rangle \nonumber\\
= \text{tr}\left\{ P _{Y_{n}}\left[ G_{n}\right] \cdots P _{Y_{1}}\left[ G_{1}\right]
\, \rho \, P _{Y_{1}}\left[ G_{1}\right] \cdots P _{Y_{n}}\left[ G_{n}\right]
\right\} ,
\label{eq:pyramid}
\end{eqnarray}
where $\rho $ is the initial state. Note that this takes the \textbf{pyramidal} form .

Here the $Z_{k}$ are understood as observables specified in the
Schr\"{o}dinger picture at time 0 - what we measure are the $Y_{k}$ which
are the $Z_{k}$ at respective times $t_{k}$. The answer depends on the
chronological order $t_{1}<t_{2}<\cdots <t_{n}$.

It is tempting to think of $\left\{ Y_{k}:k=1,\cdots ,n\right\} $ as a
discrete time stochastic process, but some caution is necessary. We cannot
generally permute the events $Y_{1}\in G_{1};\cdots ;Y_{n}\in G_{n}$ so we
do not have the symmetry usually associated with Kolmogorov's Reconstruction
Theorem.

\begin{proposition}
The finite-dimensional distributions satisfy the marginal consistency for
the most recent variable. Specifically, this means that 
\begin{eqnarray}
\sum_{G}\Pr \left\{ Y_{1}\in G_{1};\cdots ;Y_{n}\in G\right\} = \Pr \left\{
Y_{1}\in G_{1};\cdots ;Y_{n-1}\in G_{n-1}\right\}
\end{eqnarray}
where the sum is over a collectively exhaustive mutually exclusive set $%
\left\{ G\right\} $.
\end{proposition}

\begin{proof}
We have 
\begin{eqnarray}
&&\sum_{G}\Pr \left\{ Y_{1}\in G_{1};\cdots ;Y_{n}\in G\right\} \\
&=&\sum_{G}tr\big\{ P _{Y_{n-1}}\left[ G_{n-1}\right] \cdots P _{Y_{1}}\left[
G_{1}\right] \rho P _{Y_{1}}\left[ G_{1}\right] \cdots P _{Y_{n-1}}\left[
G_{n-1}\right] P _{Y_{n}}\left[ G\right] \big\}
\end{eqnarray}
but $\sum_{G}P _{Y_{n}}\left[ G\right] \equiv \hbox{1\kern-4truept 1}$ so we
obtain the desired reduction.
\end{proof}

As an example, suppose that $\{ A_k \}$ is a collection of quantum events
occurring at a fixed time $t_1$ which are mutually exclusive (that is, their
projections project onto orthogonal subspaces). Their union $\cup_k A_k$
makes sense and corresponds to the projection onto the direct sum of these
subspaces. Now if $B$ is an event at a later time $t_2$, then typically $%
\sum_k \Pr \{ A_k ; B \} $ is not the same as $\Pr \{ \cup_k A_k ; B \}$.
The well known two-slit experiment fits into this description, with $A_k$
the event that an electron goes through slit $k$ ($k=1,2$), and $B$ the
subsequent event that it hits a detector.

Marginal consistency is a property we take for granted in classical
stochastic processes and is an essential requirement for Kolmogorov's
Reconstruction Theorem. However, in the quantum setting it is only
guaranteed to work for the last measured observable. For instance, it may
not apply to $Y_{n-1}$ unless we can commute its projections with $P _{Y_{n}}%
\left[ G_{n}\right] $. The most recent measured observable has the potential to
\textbf{demolish} all the measurements beforehand.

\subsection{Bell's Inequalities}

A Bell inequality is any constraint that
applies to classical probabilities but which may fail in quantum probability.
We look at one example due to Eugene Wigner.

\begin{proposition}[A Bell Inequality]
Given three (classical) events $A,B,C$ then we always
have 
\begin{eqnarray}
\Pr \left\{ A,C\right\} \leq \Pr \left\{ A,\overline{B}\right\} +\Pr \left\{
B,C\right\} .  \label{eq:Bell}
\end{eqnarray}
\end{proposition}

\begin{proof}
From the marginal property we have the following three classical identities
$\Pr \left\{ A,\overline{B}\right\} =\Pr \left\{ A,\overline{B},C\right\}
+\Pr \left\{ A,\overline{B},\overline{C}\right\}$,  
$\Pr \left\{ B,C\right\} =\Pr \left\{ A, B,C\right\} +\Pr \left\{ \overline{%
A},B,C\right\} $,  and $\Pr \left\{ A,C\right\} =\Pr \left\{ A,B,C\right\} +\Pr \left\{ A,%
\overline{B},C\right\} $.  
We therefore have that 
\begin{eqnarray}
\Pr \left\{ A,\overline{B}\right\} +\Pr \left\{ B,C \right\} -\Pr \left\{
A,C\right\}= \Pr \left\{ A,\overline{B},\overline{C}\right\} +\Pr \left\{ 
\overline{A},B,C\right\} \geq 0.
\end{eqnarray}
\end{proof}

The proof relies on the fact marginal consistency is always valid for
classical events. Suppose that the $A,B,C$ are quantum events, say $Y_{k}$
taking a value in $G_{k}$ at time $t_{k}$ $\left( k=1,2,3\right) $, and
chronologically ordered $\left( t_{1}<t_{2}<t_{3}\right) $. Then only the
first of the three classical identities is guaranteed in quantum theory (marginal consistency
for the latest event at time $t_{3}$ only). The remaining two may fail, and it is
easy to construct a quantum system where inequality (\ref{eq:Bell}) is
violated.

\section{Quantum Filtering}
Given the issues raised above, one may ask whether it is actually possible to track a quantum system over time?

We shall say that the process $\left\{ Y_{k}:k=1,\cdots ,n\right\} $ is 
\textbf{essentially classical} whenever all the observables are compatible.
In this case its finite dimensional distributions satisfy all the
requirements of Kolmogorov's Theorem and so we can model it as a classical
stochastic process. The von Neumann algebra $\mathfrak{Y}_n$ they generate will be commutative,
and we have $\mathfrak{Y}_n \subset \mathfrak{Y}_{n+1}$ (a filtration of von Neumann algebras!).
 
For the tracking over time to be meaningful, we ask for the observed process to be essentially classical - this is the
\textbf{self-non-demolition property} of the observations.

Let $X_n= J_{(0,t_n)} (X)$ be an observable $X$ at time $t_n$. Suppose that it is compatible with $\mathfrak{Y_n}$ then we 
are lead to a well-defined classical joint distribution $\mathbb{K}_{X_n,Y_1, \cdots, Y_n}$ from which we may compute
the conditional probability $\mathbb{K}_{X_n|Y_1, \cdots, Y_n}$. This means that $\pi_{t_n} (X) = \mathfrak{E} [ X_n |
\mathfrak{Y}_n]$ is well-defined.

We only try to condition those observables that are compatible with the measured observables! This is known as the \textbf{non-demolition property}.

\subsection{Conditioning in Quantum Theory}
The natural analogue of conditional expectation in quantum theory would be a projective morphism (CP map) from a von Neumann algebra into a sub-von Neumann algebra. However, this does not to always exist in the noncommutative case. Given our discussions above, this is not surprising.

In general, be $\mathfrak{Y}$ be a sub-von Neumann algebra of a von Neumann algebra $\mathfrak{B}$, then its commutant is the set of all operators in $\mathfrak{B}$ which commute with each element of $\mathfrak{Y}$, that is
\[
\mathfrak{Y}^\prime = \{ X \in \mathfrak{B} : [X,Y]=0 , \, \forall \, Y \in \mathfrak{Y} \}.
\]

Why this is possible is easy to explain. The von Neumann algebra generated by $\mathfrak{Y}$ and fixed element $X\in \mathfrak{Y}^\prime$ will again be a commutative, so the conditional expectation of $X$ onto $\mathfrak{Y}$ is well defined. Physically, it means that $X\in \mathfrak{Y}^\prime$ is compatible with $\mathfrak{Y}$ so standard classical probabilistic constructs are valid.

As an illustration, suppose that the von Neumann algebra of a system of interest is $\mathfrak{A}$ and its environment's is $\mathfrak{E}$.
Let $X$ be a simple observable of the system, say with spectral decomposition $X = \sum_x x R_x$ where $\{ R_x \}$ is a set of mutually orthogonal projections. Similarly, let $Z = \sum_y y \, Q_y$ be an environment observable. We entangle the system and environment with a unitary $U$
(acting on $\mathfrak{A} \otimes \mathfrak{E}$) and measure the observable $Y = U^\ast
( \hbox{1\kern-4truept 1}_{\mathfrak{A}} \otimes Z ) U \equiv \sum_y P_y$. Here, the event corresponding to measuring eigenvalue $y$ of $Y$ is $P_y = U^\ast
( \hbox{1\kern-4truept 1}_{\mathfrak{A}} \otimes Q_y) U$. The observables $X_1 =
 U^\ast
( X \otimes \hbox{1\kern-4truept 1}_{\mathfrak{E}} ) U$ and $Y$ commute and so they have a well-defined joint probability to have eigenvalues $x$ and $y$ respectively for a fixed state $\mathbb{E}$:
$p(x,y) = \mathbb{E} \big[ U^\ast ( R_x \otimes Q_y ) U \big] $.
We may think of $X_1$ as $X$ evolved by one time step. Its conditional expectation given the measurement of $Y$ is then
$\mathbb{E} [ X_1 | \mathfrak{Y} ] 
= \sum_{x,y} x \, p(x|y) \, P_y $
where $p(x|y) = p(x,y) / p(y)$, with $p(y) = \sum_x p(x,y)$.

\section{Summary and Future Directions}
Quantum theory can be described in a systematic manner, despite the frequently sensationalist presentations on the subject.
Getting the correct mathematical setting allows one to develop an operational approach to quantum measurements and probabilities.
We described the quantum probabilistic framework which uses von Neumann algebras in place of measurable functions and shown how some of the usual concepts from Kolmogorov's theory carry over. However, we were also able to highlight which features of the classical world may fail to be valid in quantum theory.

What is of interest to control theorists is that the extraction of information from quantum measurements can be addressed and, under appropriate conditions, filtering can be formulated.
As we have seen, measuring quantum systems over time is problematic. However, the self-non-demolition property of the observations, and the non-demolition principle are conditions which guarantee that the filtering problem is well-posed. These conditions are met in continuous time quantum Markovian models (by causality, prediction turns out to be well-defined, though not necessarily smoothing!).
Explicit forms for the filter were given by Belavkin \cite{Belavkin}, see also \cite{BvHJ}.

\section{Cross References}
Article by H.I. Nurdin

\section{Recommended Reading}
The mathematical program of ``quantizing'' probability emerged in the 1970's and has produced a number of technical results used by physicists. However, it the prospect technological applications that has seen QP adopted as the natural framework for quantum analogues of engineering such as filtering and control. The philosophy of QP is given in \cite{Maassen}, for instance, with the main tool - quantum Ito calculus - in \cite{HP84,Par92}. The theory of quantum filtering was pioneered by V.P. Belavkin \cite{Belavkin} with modern accounts in
\cite{vH_thesis,BvHJ}.

\end{document}